\documentclass[11pt]{article}
\usepackage{graphicx}

\usepackage{amssymb,amsmath,amsthm, bbm, mathdots}
\usepackage{color, xcolor}
\usepackage{graphicx}
\usepackage{setspace}
\usepackage{xspace}
\usepackage{multirow}
\usepackage{array}
\usepackage[margin = 1in]{geometry}
\usepackage{complexity}

\usepackage{dirtytalk}
\usepackage{tikz}
\usepackage{hyperref}
\usepackage[capitalize]{cleveref}

\usepackage[shortlabels]{enumitem}
\usepackage{sectionbreak}
\sectionbreakmark{\rule{10em}{3pt}}

\hypersetup{colorlinks=true,urlcolor=blue,linkcolor=magenta,citecolor=[rgb]{.42,.56,.14},}

\usepackage{sectsty,indentfirst}

\theoremstyle{plain}
\newtheorem{theorem}{Theorem}
\newtheorem{corollary}[theorem]{Corollary}

\newtheorem{lemma}[theorem]{Lemma}
\newtheorem{claim}[theorem]{Claim}
\newtheorem{fact}[theorem]{Fact}

\newtheorem{obs}[theorem]{Observation}

\newtheorem{definition}[theorem]{Definition}

\theoremstyle{remark}

\theoremstyle{plain}
\newclass{\DNF}{DNF}
\newclass{\DNFs}{DNFs}
\newclass{\ACzero}{AC^0}
\newclass{\TCzero}{TC^0}

\renewcommand{\R}{\mathbb{R}} 
\newcommand{\F}{\mathbb{F}} 

\renewcommand{\Pr}{\mathop{\bf Pr\/}}
\renewcommand{\E}{\mathop{\bf E\/}}
\renewcommand{\D}{\mathop{\bf D\/}}

\newcommand{\abs}[1]{\left|#1\right|}

\newfunc{\MAJ}{MAJ}
\newfunc{\MUX}{MUX}
\newfunc{\NAE}{NAE}
\renewcommand{\IP}{\mathrm{IP}} 
\newfunc{\OR}{OR} 
\newfunc{\AND}{AND}
\newfunc{\Tribes}{Tribes}
\newfunc{\LocalCorrect}{LocalCorrect}

\newfunc{\sgn}{sgn} 

\newfunc{\spar}{sparsity}
\newfunc{\rank}{rank}
\newfunc{\spn}{span}
\newfunc{\quasipoly}{quasipoly}
\newfunc{\Bias}{Bias}

\newfunc{\DT}{DTdepth} 
\newfunc{\DTs}{DTsize} 

\renewcommand{\tilde}{\widetilde}

\newcommand{\calV}{\mathcal{V}}

\newcommand{\calX}{\mathcal{X}}
\newcommand{\calY}{\mathcal{Y}}

\newfunc{\Parity}{PARITY}

\newfunc{\Dict}{Dict}
\newfunc{\Corr}{Corr}
\newfunc{\avg}{avg}
\newfunc{\smooth}{smooth}
\newfunc{\dist}{dist}

\newclass{\ETH}{ETH}

\renewcommand{\epsilon}{\varepsilon}

\newcommand{\rk}{\mathsf{rk}}

\newcommand{\supp}{\mathrm{supp}}

\title{Lifting for Arbitrary Gadgets}
\author{Siddharth Iyer\thanks{Supported by NSF award 2131899.} \\ siyer@cs.washington.edu}

\providecommand{\Dfrac}[2]{\begin{array}{c} #1 \\ \hline  \hline  #2\end{array}}

\begin{document}
\maketitle 
\begin{abstract}
    We prove a sensitivity-to-communication lifting theorem for arbitrary gadgets. Given functions $f: \{0,1\}^n\to \{0,1\}$ and $g : \calX\times \calY\to \{0,1\}$, denote $f\circ g(x,y) := f(g(x_1,y_1),\ldots,g(x_n,y_n))$. 
    We show that for any $f$ with sensitivity $s$ and any $g$,
    \[D(f\circ g) \geq s\cdot \bigg(\frac{\Omega(D(g))}{\log\rk(g)} - \log\rk(g)\bigg),\]
    where $D(\cdot)$ denotes the deterministic communication complexity and $\rk(g)$ is the rank of the matrix associated with $g$.     
    As a corollary, we get that if $D(g)$ is a sufficiently large constant, 
    $D(f\circ g) = \Omega(\min\{s,d\}\cdot \sqrt{D(g)})$, where $s$ and $d$  denote the sensitivity and degree of $f$. In particular, computing the OR of $n$ copies of $g$  requires $\Omega(n\cdot\sqrt{D(g)})$ bits.
\end{abstract}

\section{Introduction}
Given two functions $f$ and $g$, how much harder is it to compute their composition than it is to compute each of the functions? 
In this work, we study this question in the model of deterministic communication complexity (see \cite{RY,KN} for a detailed reference on this topic). 
For a function $g:\calX\times \calY \to \{0,1\}$, let $D(g)$ denote the deterministic communication complexity of $g$. Given a function $g(x,y)$ and a function $f:\{0,1\}^n\to \{0,1\}$ we will be interested in the deterministic communication complexity of 
\[f\circ g(x,y) := f(g(x_1,y_1),\ldots,g(x_n,y_n)).\]

Besides being a natural problem to study, understanding function compositions especially in the context of deterministic communication, has connections to circuit lower bounds \cite{KRW}. 
We begin by recalling some prior work on related questions.

\subparagraph*{Connections to Direct Sums and XOR Lemmas.} Feder, Kushilevitz, Naor and Nisan \cite{FKNN} were the first to consider the direct sum problem for deterministic communication, which asks for lower bounds on the communication required to compute $n$ copies of $g$, denoted $g^n(x,y) := g(x_1,y_1),\ldots,g(x_n,y_n)$. Before we state their result, we define the cover number, $C(g)$, to be the smallest number of rectangles needed to cover $\calX\times \calY$ such that each rectangle is constant for $g$. Feder et al. showed the following.

\begin{theorem}[\cite{FKNN}]\label{thm:direct-sum}
    $D(g^n) \geq \log C(g^n) \geq n\cdot (\sqrt{D(g)} - \log\log(|\calX|\cdot|\calY|) - 1)$.
\end{theorem}

This was the best known bound for the deterministic communication complexity of $g^n$ until our recent work with Rao \cite{IR24b}, where we gave an XOR lemma: a lower bound for the deterministic communication complexity of $\oplus\circ g := g(x_1,y_1)\oplus \ldots \oplus g(x_n,y_n)$.
To state that result, we need the concept of rank in the context of communication complexity. For a function $g(x,y)$, let $M_g$ denote the matrix with a row for each $x\in \calX$, a column for each $y\in \calY$, and whose $(x,y)$-th entry is simply $g(x,y)$. We use the notation $\rk(g)$ to denote the rank of the matrix $M_g$. 

One might be tempted to guess that $D(\oplus \circ g) = \Omega(n\cdot D(g))$, i.e. one cannot do much better than computing $g$ in each coordinate followed by taking the parity of all the bits. However, this is false in general; when $g$ is itself the parity function on two bits, it is possible to compute $\oplus\circ g$ by simply exchanging 2 bits: $x_1\oplus\ldots\oplus x_n$ and $y_1\oplus\ldots\oplus y_n$. In \cite{IR24b}, we showed that if $g$ requires sufficiently large communication, then computing $\oplus\circ g$ requires much larger communication.

\begin{theorem}[\cite{IR24b}]\label{thm:xor-lemma}
    There exists $c_0 \geq 1$ such that for any function $g$ with $D(g) \geq c_0$, $D(\oplus\circ g) \geq \log C(\oplus\circ g) \geq n\cdot\Big(\frac{\Omega(D(g))}{\log\rk(g)} - \log\rk(g)\Big)$.
\end{theorem}

Rank and communication are believed to be closely related. For any function $g$, we know that $D(g) \geq \log \rk(g)$ and the log-rank conjecture, due to Lovász and Saks \cite{LS} asserts that $D(g) \leq (\log\rk(g))^{O(1)}$. Sudakov and Tomon \cite{ST} (see also \cite{Lovett} who gave a slightly weaker bound) recently showed that $D(g) \leq O(\sqrt{\rk(g)})$. When $D(g) = \Omega(\log^2\rk(g))$, \Cref{thm:xor-lemma} implies $D(\oplus\circ  g) \geq \Omega(n\cdot\sqrt{D(g)})$. 
Yang \cite{Yang} observed that $\rk(\oplus\circ g) \geq (\rk(g) - 1)^n - 1$ since rank tensorizes. He then concluded that if $D(g) \ll \log^2\rk(g)$, \[D(\oplus \circ g) \geq \log\rk(\oplus \circ g) \geq \Omega(n\cdot \log\rk(g))\geq  \Omega(n\cdot \sqrt{D(g)}).\] Therefore, for all $g$ with sufficiently large communication complexity, $D(\oplus\circ g)  = \Omega(n\cdot \sqrt{D(g)})$.

In this work, we generalize \Cref{thm:xor-lemma} to all functions $f$ with large sensitivity, which is defined as follows. The sensitivity of $f$ at $z\in \{0,1\}^n$ is  
\[\mathsf{s}_z(f) := |\{i: f(z) \neq f(z_1,\ldots,z_{i-1},1-z_i,z_{i+1},\ldots,z_n)\}|,\] and the sensitivity of $f$, is given by $\mathsf{s}(f) := \max_z \mathsf{s}_z(f)$. Our main result is the following.

\begin{theorem}[Main Theorem]\label{thm:main}
    There exists $c_0 \geq 1$ such that for any function $g$ with $D(g) \geq c_0$ and any function $f: \{0, 1\}^n \to \{0,1\}$, 
    \[D(f\circ g) \geq \log C(f\circ g) \geq \mathsf{s}(f)\cdot\bigg( \frac{\Omega(D(g))}{\log \rk(g)} - \log \rk(g)\bigg).\] 
\end{theorem}

Several basic functions, such as parity, the AND function, the OR function etc. all have sensitivity $n$, and the above result gives a statement analogous to \Cref{thm:xor-lemma} for such functions. Understanding the complexity of computing $f\circ g$ has also received significant attention in the area of query-to-communication lifting. Here, one typically fixes an inner function $g$, and asks for the communication complexity of computing $f\circ g$ for all $f$. We discuss the connection between our work and lifting theorems in more detail. 

\subparagraph*{Connections to Lifting Theorems. } We begin by noting that when $f$ only depends on a small number of coordinates, $f\circ g$ can easily be computed by computing $g$ in the appropriate coordinates. More generally, if one has a short decision tree (see \Cref{sec:prelim} for a definition) for $f$, then $D(f\circ g)$ can be computed by tracing out a root-to-leaf path in the decision tree and communicating to only compute $g$ in the coordinates corresponding to the path. Let the \emph{decision tree complexity} of $f$, denoted $\mathsf{DT}(f)$, be the least depth of a decision tree among those that compute $f$. Using the best decision tree for $f$, we get $D(f\circ g) \leq \mathsf{DT}(f)\cdot D(g)$. 


A natural question is whether or not the above bound is optimal. There is a large body of literature showing that the above bound is indeed optimal for several \emph{gadgets} $g$. The earliest such result is due to Raz and McKenzie \cite{RM} who considered the index function gadget, $\mathsf{Ind}_m:[m]\times \{0,1\}^m \to \{0,1\}$ given by $\mathsf{Ind}_m(x,y) = y_x$. They showed that for $m = n^{O(1)}$, and any function $f:\{0,1\}^n\to \{0,1\}$, $D(f\circ \mathsf{Ind}_m) = \Theta(\mathsf{DT}(f)\cdot D(\mathsf{Ind}_m))$. 
Their proof was simplified by Göös, Pitassi, and Watson \cite{GPW}, who also used the lifting paradigm to exhibit functions with rank $r$ and communication complexity $\tilde\Omega(\log^2 r)$. The result of \cite{RM} was recently improved by \cite{LMMPZ}, they showed the same result as that of \cite{RM} for $m = O(n\log n)$. 
Chattopadhyay, Kouck\'y, Loff, and Mukhopadhyay \cite{CKLM} showed a similar result for the inner product gadget (as well as for any gadget with a certain pseudorandom property)  $\mathsf{IP}_m : \{0,1\}^m\times\{0,1\}^m\to \{0,1\}$ given by $\mathsf{IP}_m(x,y) = x_1y_1\oplus\ldots\oplus x_my_m$. They showed that for $m = \Omega(\log n)$ and any function $f$, $D(f\circ \mathsf{IP}_m) = \Theta(\mathsf{DT}(f)\cdot D(\mathsf{IP}_m))$. Manor and Meir \cite{MM} proved that $D(f\circ g) = \Omega(\mathsf{DT}(f)\cdot D(g))$ for all functions $g$ with discrepancy at most $n^{- O(1)}$. Lifting theorems have also been studied for various other complexity measures, such as decision tree complexity to randomized communication \cite{GPWrand, CFKMP, MM}, resolution width to DAG-like communication \cite{GGKS,LMMPZ}, approximate degree to approximate rank \cite{Sherstovlifting,PR} etc.  

In light of the preceding discussion, \Cref{thm:main} can be seen as a statement that lifts sensitivity to deterministic communication complexity for any gadget $g$ whose communication complexity is a large enough constant. We also note that ours is not the first work to lift a Boolean function complexity measure other than decision tree complexity to communication. The work of Sherstov \cite{Sherstovlifting} yields randomized communication lower bounds for compositions with the index gadget using approximate degree as the complexity measure for the outer function. 

\subparagraph*{Related Boolean Function Complexity Measures. }Sensitivity and decision tree complexity are two among a few well-studied Boolean function complexity measures, such as certificate complexity, block-sensitivity, and degree; we refer the reader to the survey of Buhrman and de Wolf \cite{BdW} for a detailed reference on this topic. We use the connections between one of these complexity measures, namely the degree, to derive a lower bound on $D(f\circ g)$ solely in terms of $D(g)$ and $\mathsf{DT}(f)$. The degree of $f$, denoted $\mathsf{deg}(f)$ is the degree of the unique real, multilinear polynomial computing $f$. 
%
%
The above complexity measures are known to be related to each other up to polynomial factors. In particular, we know that for any $f$, 
\begin{align}
     \mathsf{deg}(f), \mathsf{s}(f) &\leq \mathsf{DT}(f), \label{eq:dt-all}\\
    \sqrt{\mathsf{deg}(f)} \leq \mathsf{s}(f) &\leq 2\cdot \mathsf{deg}(f)^2, \label{eq:deg-sens} \text{ and }\\
    \mathsf{DT}(f) &\leq 2\cdot \mathsf{deg}(f)^3. \label{eq:deg-DT}
\end{align}

In the preceding facts, \Cref{eq:dt-all} is due to Nisan and Szegedy \cite{NS} (see also \cite{Nisan}). The lower bound in \Cref{eq:deg-sens} is due to Huang \cite{Huang}, and the upper bound is again due to Nisan and Szegedy \cite{NS}. Lastly, \Cref{eq:deg-DT} was shown by Midrijānis \cite{Mid}. Next, we discuss our results and use the connections between the above complexity measure to derive lower bounds on the communication complexity of computing $f\circ g$.


\Cref{thm:main} together with Yang's observation yields the following lower bound on $D(f\circ g)$ in terms of the sensitivity and degree.

\begin{corollary}\label{cor}
    There exists $c_0 \geq 1$ such that for any function $g$ with $D(g) \geq c_0$ and any function $f: \{0, 1\}^n \to \{0,1\}$, 
    \[D(f\circ g) \geq \frac{\mathsf{s}(f)\cdot\mathsf{deg}(f)}{2\mathsf{s}(f) + \mathsf{deg}(f)}\cdot \Omega\bigg(\frac{D(g)}{\log\rk(g)}+ \log\rk(g)\bigg).\]
\end{corollary}

Using the fact that $2\mathsf{s}(f) + \mathsf{deg}(f) \leq 3\max\{\mathsf{s}(f),\mathsf{deg}(f)\}$ and the AM-GM inequality, we get the  following bound on $D(f\circ g)$ from \Cref{cor} 
\begin{align}\label{eq:rank-ind-lb}
D(f\circ g) \geq \Omega\big(\min\{\mathsf{s}(f),\mathsf{deg}(f)\}\cdot\sqrt{D(g)}\big).
\end{align}

Combining \Cref{eq:rank-ind-lb} with \Cref{eq:deg-sens,eq:deg-DT} we get that for any $g$ with sufficiently large communication complexity and any $f$, \[D(f\circ g) \geq \Omega(\mathsf{DT}(f)^{1/6}\cdot\sqrt{D(g)}).\]

\subparagraph*{Other Related Work. } Several works have studied the complexity of repeated computation. In communication complexity, this began with the influential proof of the randomized communication lower bound for the disjointness function \cite{KalyanasundaramSchnitger87, Razborov92}, which is the composition of the OR function on $n$ bits with the AND gadget. Their techniques were modified to prove direct sum type statements \cite{CSWY, JRS, HJMR} and streaming lower bounds \cite{BJKS}. Around the same time Raz \cite{Raz} proved the parallel repetition theorem, which was simplified by Holenstein \cite{Holenstein} and Rao \cite{Rao}. Barak, Braverman, Chen and Rao \cite{BBCR} adapted the techniques developed for the parallel repetition theorem  to prove general direct sum statements for randomized communication complexity. Several aspects of this were improved over the last decade, for example \cite{BRWY} gave a direct product theorem and \cite{Kol,Sherstov} gave near-optimal distributional direct-sum statements for product distributions. Recently, Yu \cite{Yu} gave an XOR lemma for bounded-round randomized communication, which we improved in our work with Rao \cite{IR24a} to obtain an XOR lemma for general randomized communication.

\subparagraph*{Techniques.} The following lemma is key to the proof of \Cref{thm:main}.
\begin{lemma}\label{lem:dense-rectangle}
If $f \circ g$ can be covered with $2^{T}$ monochromatic rectangles, then $g$ contains a monochromatic rectangle of density $2^{-2T/\mathsf{s}(f)}\cdot (4\cdot\rk(g))^{-2}$.
\end{lemma}

A similar statement was also derived in \cite{IR24b}, where $f$ was assumed to be the parity function on $n$ bits. 
We explain the high-level idea of the proof in a few lines and how it differs from the proof in \cite{IR24b}. Let $u(x,y)$ denote the uniform distribution on $\calX^n\times \calY^n$. A key observation made in \cite{IR24b} is that there exists a rectangle $R$ that is constant for $\oplus\circ g$, a coordinate $i$, inputs $x_1,\ldots,x_{i-1}$ and $y_{i+1},\ldots,y_n$ such that
\begin{enumerate}
\item\label{item:one} $|\mathsf{supp}(u(x_i,y_i\vert x_1,\ldots,x_{i-1},y_{i+1},\ldots,y_n,R))| \geq \Omega(|\calX|\cdot|\calY|\cdot2^{-T/n})$, where $\mathsf{supp}$ denotes the support of a distribution, and 
\item\label{item:two} there exists a bit $b\in \{0,1\}$ so that for every $(x',y') \in R$ consistent with $x_1,\ldots,x_{i-1}$ and $y_{i+1},\ldots,y_n$ we have $\oplus_{j\neq i} g(x_j',y_j') = b$.
\end{enumerate}

Since $R$ is constant for $\oplus \circ g$, \Cref{item:one,item:two} together imply that $g(x_i,y_i)$ is fixed for every $(x_i,y_i)$ in the support of $u(x_i,y_i\vert x_1,\ldots,x_{i-1},y_{i+1},\ldots,y_n,R)$. Moreover, since $x,y$ are independent in $u$, the support is a rectangle.

To prove lower bounds for general function composition, we need an appropriate generalization of the parity constraint in \Cref{item:two}. We enforce such a constraint by switching from the uniform distribution to a correlated distribution. 

For simplicity, assume $\mathsf{s}(f) = n$ and that
$g$ is balanced. By definition, there exists $z\in \{0,1\}^n$ with $f(z) \neq f(z_1,\ldots,z_{i-1},1-z_i,z_{i+1},\ldots,z_n)$, for all $i$. Consider the distribution $p(x,y)$ on $\calX^n\times \calY^n$ obtained by sampling each $(x_i,y_i)$ uniformly conditioned on $g(x_i,y_i) = z_i$. Now, we obtain a weaker statement than \Cref{item:one}, which suffices for our purposes. We show that there exists a rectangle $R$ that is constant for $f\circ g$, a coordinate $i$ and inputs $x_1,\ldots,x_{i-1},y_{i+1},\ldots,y_n$ such that the sets
\begin{align*}
    A &:=  \mathsf{supp}(p(x_i\vert x_1,\ldots,x_{i-1},y_{i+1},\ldots,y_n,R)) \text{ and }\\
    B &:= \mathsf{supp}(p(y_i\vert x_1,\ldots,x_{i-1},y_{i+1},\ldots,y_n,R)),
\end{align*}
satisfy $|A| \geq \Omega(|\calX|\cdot 2^{-T/n})$  and $|B| \geq \Omega(|\calY|\cdot 2^{-T/n})$.

Similar to the proof of \Cref{item:one}, this uses the sub-additivity of entropy, we refer the reader to \Cref{claim:subadditivity} for more details. Next, we show that $A\times B$ is a monochromatic rectangle for $g$. We note that \[\mathsf{supp}(p(x_i,y_i\vert x_{1},\ldots,x_{i-1},y_{i+1},\ldots,y_n,R)) \subseteq A\times B,\] and although $g$ is constant on the former set (by the definition of $p$), it is not obvious that the same holds for the $A\times B$. Nevertheless, we show that this is indeed the case. 

The intuition for this is as follows. For any $(x_i,y_i)\in A\times B$, one can use the definitions of $A,B$ and $p(x,y)$, to show that there exist $(x',y')\in R$ satisfying $x'_i=x_i$, $y'_i = y_i$ and $g(x'_j,y'_j) = z_j$, for all $j \neq i$. Therefore, if $g(x_i,y_i) \neq z_i$, then $f\circ g(x',y') \neq f(z)$, by the sensitivity of $f$. This is a contradiction because for every $(x,y)\in R \supseteq \mathsf{supp}(p(x,y\vert R))$, $f\circ g$ evaluates to $f(z)$. 

The main difference between the above high-level description and the proof of \Cref{thm:main} is that $g$ need not be balanced. To address this, we consider two cases. First, we assume that $g$ is extremely biased, say $\Pr[g(x,y) = 1] > 1 - 1/(4\cdot \rk(g))$. In this case, we obtain a monochromatic rectangle for $g$ using an observation of Gavinsky and Lovett \cite{GL}, ignoring the cover for $f\circ g$. Otherwise, $g$ is not too biased and we can apply the above discussion, albeit with a loss of $1/(4\cdot \rk(g))$ in the final bound.

\subparagraph*{Organization. } We recall relevant definitions in \Cref{sec:prelim} and prove \Cref{lem:dense-rectangle} in \Cref{sec:lemma-proof}. \Cref{thm:main} and \Cref{cor} are proved in \Cref{sec:main-thm-proof}.

\section{Preliminaries}\label{sec:prelim}
For shorthand, we denote by $[n]$ the set $\{1,\ldots,n\}$. For an element $x\in \calX^n$, we will refer to $x_1,\ldots,x_{i-1}$ by $x_{< i}$ and similarly, we will refer to $x_{i},\ldots,x_n$ by $x_{\geq i}$. Similarly, for a random variable $X$ taking values in $\calX^n$, we use $X_{< i}$ to denote $X_1,\ldots,X_{i-1}$. All logarithms will be taken base 2. Given a distribution $p(x)$ we use $\mathsf{supp}(p(x))$ to denote the set of points in the support of $p(x)$. We recall some relevant mathematical facts, which we will use later.

\begin{definition}[Entropy]
    Given a random variable $A$ distributed according $p(a)$ the entropy of $X$ is given by
    \[H(A) := \E_{p(a)}\bigg[\log \frac{1}{p(a)}\bigg].\]
\end{definition}

\begin{fact}\label{fact:entropy-upper-bound}
    If $A$ has finite support, then $H(A) \leq \log|\mathsf{supp}(p(a))|$, with equality if $p(a)$ is the uniform distribution. 
\end{fact}

Given two jointly distributed random variables $A$ and $B$ distributed according to $p(ab)$, the conditional entropy of $B$ given $A$ is defined as 
\begin{equation*}
H(B\vert A) :=  \E_{p(ab)}\bigg[\log\frac{1}{p(b\vert a)}\bigg].
\end{equation*}

It is well-known that $H(B\vert A) \leq H(B)$. We also recall the chain rule for entropy
\begin{align}\label{eq:chain-rule}
H(AB) = H(A) + H(B\vert A).
\end{align}

Next, we recall the notion of KL-divergence.

\begin{definition}[KL-divergence]
Given two probability distributions $p(a)$ and $q(a)$, the \emph{KL-divergence} between $p$ and $q$ is defined as 
\[
    \D(p||q) := \E_{p(a)}\bigg[\log \frac{p(a)}{q(a)}\bigg].
\]
\end{definition}

\begin{fact}\label{claim:KL-non-negative}
    For any two distributions $p(a)$ and $q(a)$, it holds that $ \D(p||q) \geq 0$.
\end{fact}

Next, we recall two complexity measures associated with Boolean functions: decision tree complexity and degree. A decision tree of depth $d$ is an adaptive (deterministic) query algorithm, making at most $d$ queries to compute a given function. The algorithm queries variables $x_{i_1},\ldots,x_{i_d}$ adaptively and outputs a bit based on the values of the variables. We say that a decision tree computes a function $f$, if on every input $x$, the algorithm outputs $f(x)$. The decision tree complexity of $f$, denoted $\mathsf{DT}(f)$ is the least depth of a decision tree among those that compute $f$. 
We also recall that for every function $f:\{0,1\}^n \to \{0,1\}$, there exists a unique real, multilinear polynomial \[q(x) = \sum_{S\subseteq [n]} c_S\cdot \prod_{i\in S}x_i,\] such that $q(x) = f(x)$ for all $x\in \{0,1\}^n$. The degree of $f$, denoted $\mathsf{deg}(f)$, is the degree of $q$. 

We conclude this section with some facts regarding communication complexity and its connections to the rank of matrices, whose proofs can be found in \cite{RY}. First, we recall that rank is sub-additive.

\begin{fact}\label{fact:rk-subadd}
    For two matrices $A_1$ and $A_2$, we have $\rk(A_1+A_2) \leq \rk(A_1) + \rk(A_2)$.
\end{fact}

Next, we note that the communication complexity of a function is at least the logarithm of rank of the corresponding matrix.

\begin{fact}\label{fact:log-rk}
    For any function $g:\calX\times \calY\to \{0,1\}$, we have $D(g) \geq \lceil\log\rk(g)\rceil$.
\end{fact}

Lastly, we need the fact that a protocol with a small number of leaves can be simulated by a short protocol (see \cite{RY} Theorem 1.7). 

\begin{fact}\label{fact:rebalancing}
    Given a protocol $\pi$ with $\ell$ leaves, there exists a protocol with communication at most $\lceil2 \log_{3/2}\ell\rceil$ that outputs $\pi(x,y)$ on inputs $x$ and $y$.
\end{fact}

\section{Proof of \Cref{lem:dense-rectangle}}\label{sec:lemma-proof}
    First, assume that $\abs{\E_{x,y}[g(x,y)] -1/2} > 1/2 -  1/(4\cdot \rk (f))$. In this case, we claim $g$ contains a monochromatic rectangle of constant density. Indeed, 
    \begin{align*}
        \abs{\E_{x,y}[g(x,y)] - \frac{1}{2}} = \max\bigg\{\Pr_{x,y}[g(x,y) = 1] - \frac{1}{2},\Pr_{x,y}[g(x,y) = 0] - \frac{1}{2} \bigg\},
    \end{align*}
    and we can assume without loss of generality that 
    \[\Pr_{x,y}[g(x,y) = 1] > 1 - \frac{1}{4\cdot \rk(g)}.\]
    Let $E$ be the set of all $x$ such that  $\Pr_y[g(x,y) = 1] \leq 1 - 1/(2\cdot \rk(g))$. We have,
    \begin{align*}
        \Pr_{x,y}[g(x,y) = 1] &\leq  \Pr_x[x\in E]\cdot \bigg(1 -\frac{1}{2\cdot \rk(g)}\bigg) + \Pr_x[x\notin E] \\
        &= 1 - \frac{\Pr_x[x\in E]}{2\cdot \rk(g)},
    \end{align*}
    which implies that $\Pr_{x}[x\in E] \leq 1/2$. Let $x_1,\ldots,x_r\in E^c$ be such that the corresponding rows are maximally linearly independent in $M_g$. Moreover, let $G = \{y: f(x_i,y) =1,  \forall i\in [r]\}$. By a union bound, we have \[\Pr_y[y\notin G]  \leq r\cdot \frac{1}{2\rk(g)} \leq \frac{1}{2}.\] We observe that $E^c\times G$ is a monochromatic rectangle of density at least $1/4$. Since $2^{-2T/s}\cdot(4\cdot\rk(g))^{-2} \leq 1/16 < 1/4$ we have found a monochromatic rectangle of the desired density.

    Next, suppose that $\abs{\E_{x,y}[g(x,y)] - 1/2} < 1/2 -  1/(4\cdot \rk (f))$. We have
    \begin{align}\label{eq:balanced}
        \frac{1}{4\cdot\rk(g)} \leq \Pr_{x,y}[g(x,y) = 0], \Pr_{x,y}[g(x,y) = 1] < 1 - \frac{1}{4\cdot\rk(g)}.
    \end{align}
    
    For shorthand let $s$ be the sensitivity of $g$. By definition, there exists an input $z\in \{0,1\}^n$ and a set $S\subseteq [n]$ such that 
    \[f(z) \neq f(z_{< i},1 - z_i,z_{> i}).\] 
    We may assume without loss of generality that $S \supseteq [s]$, for otherwise, this can be ensured by renaming the coordinates. Let $u(x,y)$ denote the uniform distribution over all inputs $(x,y)\in \calX^n\times \calY^n$, and let $p(x,y)$ be a distribution obtained by sampling each $(x_i,y_i)$ randomly and independently subject to $g(x_i,y_i) = z_i$. By \Cref{eq:balanced} we have the following inequality relating the two distributions:
    \begin{align}
        \max_{x,y} \frac{p(x_{\leq s},y_{\leq s})}{u(x_{\leq s},y_{\leq s})} &= \max_{xy} \prod_{i\leq s}\frac{u(x_i,y_i\vert g(x_i,y_i) = z_i)}{u(x_i,y_i)} \nonumber \\
        &= \prod_{i\leq s} \frac{1}{\Pr_{u(x_i,y_i)}[g(x_i,y_i) = z_i]} \leq  (4\cdot \rk(g))^{s}. \label{obs:max-bound}
    \end{align}

    Recall that $f\circ g$ can be covered with at most $2^T$ monochromatic rectangles, say $R_1,\ldots,R_{2^T}$. Thus, there exists a rectangle $R$ in the cover with $p(R) \geq 2^{-T}$. 
    Let $X$ and $Y$ be random variables denoting rows and columns of $\calX^n$ and $\calY^n$ respectively, where $XY$ is distributed according to $p(x,y\vert R)$.
    \begin{claim}\label{claim:subadditivity}
        \[\sum_{i\in [s]} H(X_i\vert X_{< i},X_{> s},Y_{> i}) + H(Y_i\vert X_{< i},X_{> s},Y_{> i}) \geq s\log\frac{|\calX|\cdot|\calY|}{(4\cdot\rk(g))^2} - 2T.  \]
    \end{claim}
    \begin{proof}
    Applying the chain rule for entropy we get
    \begin{align*}
        &\sum_{i\in [s]} H(X_i\vert X_{< i},X_{> s},Y_{> i}) + H(Y_i\vert X_{< i},X_{> s},Y_{> i}) \\&\geq \sum_{i\in [s]} H(X_i\vert X_{< i},X_{> s},Y) + H(Y_i\vert X,Y_{> i}) 
        = H(X\vert X_{> s},Y) + H(Y\vert X,Y_{> s}).
    \end{align*}
    Let $p'(x,y)$ be the distribution obtained by sampling $(x_j,y_j)$ uniformly at random, for each $j\in [s]$, and according to $p(x_j,y_j)$ for each $j\notin S$. Using this notation, we bound the term $H(X\vert X_{> s},Y)$ above as follows  
    \begin{align*}
    &H(X\vert Y,X_{> s}) \\&= \E_{p(x,y\vert R)}\bigg[\log\frac{1}{p(x\vert y,x_{> s},R)}\bigg] \\
    &= \E_{p(x,y\vert R)}\bigg[\log\frac{p(R)\cdot p(x_{> s},y\vert R)}{p(x,y)}\bigg] \\
    &\geq \E_{p(x,y\vert R)}\bigg[\log\frac{p(x_{> s},y\vert R)}{(4\cdot\rk(g))^s\cdot u(x_{\leq s},y_{\leq s})\cdot p(x_{> s},y_{> s})}\bigg] - T \tag{\Cref{obs:max-bound} and $p(R)\geq 2^{-T}$}\\
    &= \E_{p(x,y\vert R)}\bigg[\log\frac{|\calX|^s\cdot p(x_{> s},y\vert R)}{p'(x_{> s},y)}\bigg] - T - s\log(4\cdot\rk(g))\\
    &= s\log \frac{|\calX|}{4\cdot\rk(g)} + \D(p(x_{>s},y\vert R)||p'(x_{> s},y)) - T \geq s\log \frac{|\calX|}{4\cdot\rk(g)} -T,
    \end{align*}    
    which follows by \Cref{claim:KL-non-negative}. A similar calculation shows that $H(Y\vert X,Y_{> s}) \geq s\log \frac{|\calY|}{4\cdot\rk(g)} - T$, yielding the desired bound.
    \end{proof}
    
    By an averaging argument, we obtain an index $i\in[s]$ and $x_{< i},x_{>s},y_{> i}$ such that
    \begin{align*}
        H(X_i\vert x_{< i},x_{>s},y_{> i}) + H(Y_i\vert x_{< i},x_{>s},y_{> i}) \geq \log\frac{|\calX|\cdot|\calY|}{(4\cdot\rk(g))^2} - \frac{2T}{s}.
    \end{align*}
    For shorthand, let \[A := \mathsf{supp}(p(x_i\vert x_{<i},x_{>s},y_{>i},R)) \text{ and } B:= \mathsf{supp}(p(y_i\vert x_{<i},x_{>s},y_{>i},R)).\]
    
    Using \Cref{fact:entropy-upper-bound} we conclude that the rectangle given by $A\times B$ has size at least 
    \[\frac{|\calX|\cdot |\calY|}{16\cdot\rk(g)^2\cdot 2^{2T/s}}. \]
    
    Moreover, we claim that $A\times B$ is monochromatic for $g$.
    Indeed, for any $x_i\in A$, there exists a row $x'\in \mathsf{supp}(p(x\vert x_{< i},x_{> s},y_{> i},R))$ such that $x'_{i} = x_i$. Similarly, for any $y_i\in B$, there exists a column $y'\in \mathsf{supp}(p(y\vert x_{< i},x_{> s},y_{> i},R))$ such that $y'_i = y_i$. In particular, $(x',y')\in R$ and in addition, $x'_{j} = x_j$ for all $j< i$ and $y'_j = y_j$ for all $j> i$. 
    
    Since $y'\in \supp(p(y\vert x_{<i},x_{> s},y_{> i},R))$, we get $g(x'_t,y'_t) = g(x_t,y'_t)= z_t$ for all $t < i$. Similarly, we have $g(x'_t,y'_t) = g(x'_t,y_t)= z_t$ for all $t > i$. Since 
    $i\in [s]$, if $g(x_i',y_i') \neq z_i$, then $f\circ g(x',y') = f(z_{<i},1- z_i,z_{> i})\neq f(z)$. However, this contradicts the fact that $R$ is monochromatic for $f\circ g$ because for every $x,y \in \mathsf{supp}(p(x,y\vert R))$, we know that $f\circ g(x,y) = f(z)$. 

\section{Proof of \Cref{thm:main}}\label{sec:main-thm-proof}
    At a high level, we apply \Cref{lem:dense-rectangle} repeatedly to find dense monochromatic rectangles and combine this with the arguments of \cite{NW} to obtain a protocol for $g$. 

    Fix some two functions $f$ and $g$ and consider any cover of $f\circ g$ with some $2^T$ monochromatic rectangles. For shorthand, let $s$ denote the sensitivity of $f$. Applying \Cref{lem:dense-rectangle}, we obtain a monochromatic rectangle $R$ in $M_g$ with density $2^{-2T/s}\cdot(4\cdot\rk(g))^{-2}$. 
    
    By renaming the rows and columns of $M_g$ appropriately, we can rewrite it as
    \[\begin{bmatrix}R & A\\
    B & Z\end{bmatrix},\]
    for some matrices $A,B$ and $Z$. Now, we observe that 
    \begin{align}\label{eq:rank-decrement}
    \min\Bigg\{\rk\bigg(\begin{bmatrix}R & A\end{bmatrix}\bigg),  \rk\Bigg(\begin{bmatrix}R\\ B\end{bmatrix}\Bigg) \Bigg\}\leq  \frac{\rk(g) + 3}{2}.
    \end{align}
    Since $R$ has rank one we get 
    \begin{align*}
       \rk\bigg(\begin{bmatrix}R & A\end{bmatrix}\bigg) +   \rk\Bigg(\begin{bmatrix}R\\ B\end{bmatrix}\Bigg) 
       &\leq \rk(A) + \rk(B) +2 \tag{by \Cref{fact:rk-subadd}} \\
       &\leq \rk\Bigg(\begin{bmatrix}0 & A\\
    B & Z\end{bmatrix}\Bigg) +2 \tag{Gaussian Elimination}\\ 
    &\leq \rk\bigg(\begin{bmatrix}R & A\\
    B & Z\end{bmatrix}\bigg) + 3 \tag{by \Cref{fact:rk-subadd}}\\
    &= \rk(g) + 3, 
    \end{align*}
    and \Cref{eq:rank-decrement} follows.

    If $\rk\Big(\begin{bmatrix}R & A\end{bmatrix}\Big) \leq (\rk(g)+3)/2$, Alice sends a bit to Bob indicating whether or not her input is consistent with the rows of $R$. Otherwise, Bob sends a bit to Alice indicating whether or not his input is consistent with the columns of $R$. We can assume without loss of generality that $\rk\Big(\begin{bmatrix}R & A\end{bmatrix}\Big) \leq (\rk(g)+3)/2$ as the proof is symmetric.  
    
    Let $g'$ and $g''$ denote the functions encoded by the matrices $\begin{bmatrix}R & A\end{bmatrix}$ and  $\begin{bmatrix}B & Z\end{bmatrix}$ respectively. We note that a cover of $M_{f\circ g}$ also gives a cover of both $M_{f\circ g'}$ and $M_{f\circ g''}$.    
    If Alice's input is consistent with the rows of $R$, the players repeat the above argument using the rectangle cover for $M_{f\circ g'}$.  Otherwise, they repeat the argument using the rectangle cover for $M_{f\circ g''}$. In the former case, we have $\rk(g') \leq (\rk(g)+3)/2$ and in the latter case, the size of $\calX\times\calY$ shrinks by a factor of $1 - 2^{-2T/s}\cdot(4\cdot\rk(g))^{-2}$. 
    
    We claim that after  $(4\cdot \rk(g))^3\cdot 2^{2T/s} + O(\log \rk(g))$ recursive steps either the rank is at most $5$ or the size of the matrix is at most $1$. Indeed, as long as the $\rk(g) \geq 5$,  we have $\rk(g') \leq (\rk(g) + 3)/2 \leq 4\cdot \rk(g)/5$. Hence, there can only be  after $\log_{5/4} \rk(g)$ many steps where the rank reduces by a factor of $4/5$. Similarly, there can be only $k = (4\cdot \rk(g))^3\cdot 2^{2T/s}$ many steps where the size of the matrix reduces, since 
    \begin{align*}
        \bigg(1 - \frac{1}{2^{2T/s} (4\cdot\rk(g))^{2}}\bigg)^{k} &\leq \exp\Big(- \frac{k}{2^{2T/s}(4\cdot \rk(g))^2}\Big)
        = e^{-4\cdot\rk(g)}
        \leq \frac{1}{|\calX|\cdot |\calY|},
    \end{align*}
    where we used the fact that $|\calX|$ and $|\calY|$ are both at most $2^{\rk(g)}$ in the last step. 

    Every leaf of this protocol either corresponds to a size $1$ matrix or a matrix of rank at most $5$. Thus, with constantly more bits of communication, we get a protocol for $g$ with the following upper bound on the number of leaves:
    \begin{align*}
    \binom{(4\cdot \rk(g))^3\cdot 2^{2T/s} + O(\log \rk(g))}{O(\log \rk(g))}\cdot O(1) 
    &\leq O(\rk(g)^3\cdot 2^{2T/s})^{O(\log\rk(g))}
    \\
    &\leq 2^{O((T/s + \log\rk(g))\cdot \log\rk(g))},
    \end{align*}
    where all the inequalities hold for $c_0$ large enough.
    
    By \Cref{fact:rebalancing} the above protocol can be rebalanced to have communication at most \[O\bigg(\bigg(\frac{T}{s}+ \log\rk(g)\bigg)\cdot \log\rk(g)\bigg).\]
    Since $g$ requires communication at least $D(g)$,we have \begin{align*}
        \bigg(\frac{T}{s} + \log\rk(g)\bigg)\cdot \log\rk(g) \geq \Omega(D(g)), 
    \end{align*}
    and the theorem follows by rearranging.

    \subsection{Proof of \Cref{cor}}
    We start with the following claim relating $\rk(g)$ and $\rk(f\circ g)$.

    \begin{lemma}\label{lemma:rank-lb} For any two functions $f:\{0,1\}^n\to \{0,1\}$ and $g:\calX\times\calY\to \{0,1\}$, it holds that $\rk(f\circ g) \geq (\rk(g)-1)^{\mathsf{deg}(f)}$.
\end{lemma}
\begin{proof}
    For shorthand, denote by $d$, the degree of $f$. By definition, there exists a subset of size $d$ whose corresponding coefficient in the polynomial expansion of $f$ is non-zero. We can assume without loss of generality that this set is $[d]$, otherwise, we can rename the variables to ensure this. 
    Let $u_1,\ldots,u_{r}$ be a maximal set of linearly independent rows of $M_g$, and let $x_1,\ldots,x_r$ be the corresponding inputs. Further, define the vectors $\tilde u_1,\ldots,\tilde u_r$, where $\tilde u_i$ is the projection of $u_i$ onto the space orthogonal to the all-ones vector, $\mathbf 1$. We note that the dimension of $\mathsf{span}(\tilde u_1,\ldots,\tilde u_r)$ is at least $r-1$.

    In what follows, we adopt the following notation for the tensor product of 2 (or more vectors). Given two vectors $u\in \R^m$ and $v\in \R^k$, we denote the tensor product of $u$ with $v$ by $u\otimes v \in \R^{mk}$  where $u\otimes v[i,j] = u(i)\cdot v(j)$. 
    
    The key observation is that the projection of the rows of $M_{f\circ g}$ to the space 
    \[\calV := \mathsf{span}(\{v_1\otimes\ldots\otimes v_d\otimes \underbrace{\mathbf{1} \otimes \ldots \otimes \mathbf{1}}_{n-d \text{ times}}: v_i \in \{\tilde u_1,\ldots,\tilde u_r\}\})\]
    has full rank. Indeed, consider any function $h:[n]\to [r]$ and let $u_h$ be the row corresponding to the inputs $x_{h(1)},\ldots,x_{h(n)}$. For any $y_1,\ldots,y_n$, using the multilinear polynomial for $f$ we can write
    \begin{align*}
        u_h(y_1,\ldots,y_n) &= f(g(x_1,y_1),\ldots,f(x_n,y_n))\\
        &= \sum_{S\subseteq [n]}\alpha_S\cdot \prod_{i\in S} g(x_i,y_i) \\
        &= \sum_{S\subseteq [n]}\alpha_S\cdot \prod_{i\in S} u_{h(i)}(y_i). 
    \end{align*}
    
    For any set $S$, the last quantity above can be written as tensor product. For example, if we let $S = [t]$ then 
    \begin{align*}
    \prod_{i\in [t]} u_{h(i)}(y_i) &= u_{h(1)}\otimes \ldots u_{h(t)}[y_1,\ldots,y_t] \\
    &= u_{h(1)}\otimes \ldots u_{h(t)}\otimes \underbrace{\mathbf{1} \otimes \ldots \otimes \mathbf 1}_{n-t \text{ times}} [y_1,\ldots,y_n].
    \end{align*}
    Applying this to a general set $S$, we can write
    \[ \prod_{i\in S} u_{h(i)}(y_i) = \otimes_{i\in S} u_{h(i)}\otimes_{i\notin S} \mathbf{1}[y_1,\ldots,y_n],\]
    where the subscript is used to denote the vector in the $i$-th coordinate of the tensor product depending on whether or not $i\in S$.

    For any set $S \neq [d]$ of size at most $d$, the projection of $\otimes_{i\in S} u_{h(i)}\otimes_{i\notin S} \mathbf{1}$ onto $\calV$ is zero, since there exists $i\in [d]\setminus S$ such that the vector in the $i$-th coordinate of the tensor product is $\mathbf 1$. 
    Moreover, by the definition of degree, $\alpha_S = 0$ for sets $S$ of size larger than $d$. 
    Lastly, the projection of $\alpha_{[d]}\cdot \otimes_{i\in [d]} u_{h(i)}\otimes_{i=d+1}^n\mathbf{1}$ is exactly $\alpha_{[d]}\cdot\otimes_{i\in [d]} \tilde u_{h(i)}\otimes_{i=d+1}^n\mathbf{1}$.
    
    This establishes that the projection of the rows of $M_{f\circ g}$ to $\calV$ has full rank. It follows the rank of $M_{f\circ g}$ is at least the dimension of $\calV$, which is at least $(r-1)^{d} = (\rk(g)-1)^{\mathsf{deg}(f)}$.
\end{proof}
    
    \begin{proof}[Proof of \Cref{cor}]
        Recalling \Cref{fact:log-rk}, we have $D(f\circ g) \geq \lceil \log \rk(f\circ g)\rceil$. Therefore, we can put together \Cref{thm:main} and \Cref{lemma:rank-lb} to conclude
        \begin{align*}
            D(f\circ g) \geq \max\bigg\{\mathsf{s}(f)\cdot\bigg(\frac{\Omega(D(g))}{\log \rk(g)} - \log\rk(g)\Bigg), \mathsf{deg}(f)\cdot\log(\rk(g) -1)\bigg\}.
        \end{align*}
        Using the fact that $\max\{a,b\} \geq \lambda\cdot a + (1-\lambda)\cdot b$, for any $\lambda \in [0,1]$, we can set $\lambda = \mathsf{deg}(f)/(2\mathsf{s}(f) +\mathsf{deg}(f))$ to get
        \begin{align*}
            D(f\circ g) &\geq \frac{\mathsf{s}(f)\cdot \mathsf{deg}(f)}{2\mathsf{s}(f) + \mathsf{deg}(f)}\cdot\bigg(\frac{\Omega(D(g))}{\log\rk(g)} - \log\rk(g) + 2\log(\rk(g)-1)\bigg)\\
            &\geq \frac{\mathsf{s}(f)\cdot \mathsf{deg}(f)}{2\mathsf{s}(f) + \mathsf{deg}(f)}\cdot\Omega\bigg(\frac{D(g)}{\log\rk(g)} + \log\rk(g)\bigg),
        \end{align*}
        where we used the fact that for $c_0$ large enough, $(\rk(g) -1)^2 \geq \rk(g)^{3/2}$.
    \end{proof}

    \section{Concluding Remarks}
    At a high level, we try to understand the communication complexity of computing $f\circ g$ for arbitrary $f$ and $g$. We expect that this is $\Omega(\mathsf{DT}(f)\cdot D(g))$, for any sufficiently complex gadget $g$. \Cref{cor} gives the lower bound $\Omega(\min\{\mathsf{s}(f),\mathsf{deg}(f)\}\cdot \sqrt{D(g)}) \geq \Omega(\mathsf{DT}(f)^{1/6}\cdot \sqrt{D(g)})$ which can be seen as progress towards this. Below, we show that for some gadgets this can be further improved. 
    \begin{enumerate}
        \item[1. ] For certain gadgets $g$, we can obtain $D(f\circ g)\geq \Omega(\mathsf{DT}(f)^{1/3}\cdot \sqrt{D(g)})$ using the notion of \emph{block-sensitivity}, a well-studied \cite{NS} Boolean function complexity measure. The block-sensitivity of $f$ at $z\in \{0,1\}^n$ is the maximum number of disjoint sets $S_1,\ldots,S_t$ such that for all $i\in [t]$, \[f(z) \neq f(z^{\oplus S_i}), \text{ where } z^{\oplus S_i}_j = \begin{cases} 1 - z_j, &\text{if } j\in S_i \text{ and}\\ z_j, &\text{ otherwise.}\end{cases} \]
        The block-sensitivity of $f$, denoted $\mathsf{bs}(f)$, is the maximum across all $z$, of the block-sensitivity of $f$ at $z$. By definition, $\mathsf{s}(f) \leq \mathsf{bs}(f)$. Moreover, the block-sensitivity is known to give a better upper bound for the decision-tree complexity than the sensitivity. In particular, Midrijānis \cite{Mid} showed that $\mathsf{DT}(f) \leq \mathsf{bs}(f)\cdot \mathsf{deg}(f)$. By the lower bound in \Cref{eq:deg-sens} we know that $\mathsf{deg}(f) \leq \mathsf{s}(f)^2 \leq \mathsf{bs}(f)^2$. Hence, $\mathsf{DT}(f) \leq \mathsf{bs}(f)^3$. 
        
        For any function $g$ with the following symmetry property, one can replace $\mathsf{s}(f)$ in \Cref{thm:main} with $\mathsf{bs}(f)$, which together with \Cref{cor} implies that $D(f\circ g) \geq \Omega(\mathsf{DT}(f)^{1/3}\cdot \sqrt{D(g)})$. The symmetry property in question is that for any $x$, there exists $\overline x\neq x$ such that $g(x,y) = 1 - g(\overline x,y)$, for every $y$. An example of such a function is the index function, $\mathsf{Ind}_m$. We give a sketch of this claim. Suppose $f$ has block-sensitivity $b$, achieved at a point $\tilde z$ by sets $S_1,\ldots,S_b$. Consider the function $f' :\{0,1\}^b\to\{0,1\}$ given by \[f'(z) = f(z'), \text{ where } z'_j = \begin{cases} \abs{\tilde z_j - z_j}, &\text{if } j\in S_1\cup\ldots\cup S_b \text{ and}\\ \tilde z_j, &\text{ otherwise.}\end{cases}\] 
        We note that $f'$ has sensitivity $b$, since $f'(0) = f(\tilde z)$, and $f'(0^{\oplus \{i\}}) = f(\tilde z^{\oplus S_i})$. Moreover, any protocol that computes $f\circ g$ can also be used to compute $f'\circ g$ in the following way. Suppose Alice and Bob gets inputs $x_1,\ldots,x_b$ and $y_1,\ldots,y_b$. For each set $S_i$ and coordinate $j\in S_i$, Alice sets $x_j = x_i$ if $\tilde z_j = 0$, and otherwise, sets $x_j = \overline x_i$ (from the symmetry property). Bob sets $y_j = y_i$ for each $j\in S_i$. For every coordinate $j\notin S_1\cup\ldots\cup S_b$, the players arbitrarily fix inputs such that $g(x_j,y_j) = \tilde z_j$. They can now run the protocol for $f\circ g$ to compute $f'\circ g$ and it follows that $D(f'\circ g) \leq D(f\circ g)$. However, \Cref{thm:main} shows that \[D(f'\circ g) \geq \mathsf{s}(f')
        \cdot \bigg(\frac{\Omega(D(g))}{\log\rk(g)} - \log\rk(g)\bigg) = \mathsf{bs}(f)
        \cdot \bigg(\frac{\Omega(D(g))}{\log\rk(g)} - \log\rk(g)\bigg).\]
        \item[2. ] Anup Rao observed that for certain other gadgets $g$, such as the inner product gadget $\mathsf{IP}_m$, one can improve \Cref{thm:main} to obtain $D(f\circ g) = \Omega(\mathsf{s}(f)\cdot D(g)/\log\rk_2(g))$, where $\rk_2(g)$ is the rank of $M_g$ over $\F_2$. This can be seen by modifying the proof of \Cref{thm:main} to keep track of $\rk_2(g)$ instead of $\rk(g)$. Each time we find a monochromatic rectangle $R$ for $g$ using \Cref{lem:dense-rectangle}, we can recurse on a sub-matrix where either $\rk_2(g)$ goes down by a factor of $4/5$ or the size of the matrix shrinks by the appropriate amount. If the rank over $\F_2$ is at most 5, one can just use 6 bits of communication to compute the function since $D(g) \leq \rk_2(g) + 1$. This calcuation yields
        \[D(f\circ g) = \mathsf{s}(f)
        \cdot \bigg(\frac{\Omega(D(g))}{\log\rk_2(g)} - \log\rk_2(g)\bigg).\]
        Furthermore, for any gadget satisfying $D(g) = \Omega(\log^2\rk(g))$, we get that $D(f\circ g) = \Omega(\mathsf{s}(f)
        \cdot D(g)/\log\rk_2(g))$. In particular, for the innner product gadget, we know that $\rk_2(\IP_m) \leq m$ and $D(\IP_m) = \Omega(m)$.
    \end{enumerate}

    \subparagraph*{Acknowledgments.} We are grateful to Anup Rao for several helpful discussions and his encouragement to work on the problem. We would like to thank Guangxu Yang for letting us know that \Cref{thm:xor-lemma} can be strengthened; the proof of \Cref{cor} uses his idea. We are also grateful to Paul Beame, Oscar Sprumont and Michael Whitmeyer for helpful conversations and feedback on a draft of this work.
\bibliographystyle{alpha}
\bibliography{ref}
\end{document}